\documentclass[11pt,a4paper]{article}

\usepackage[utf8]{inputenc} % These source files are coded in UTF-8.

\usepackage[linesnumbered,ruled,vlined]{algorithm2e} \DontPrintSemicolon % In order to use, for example, "\begin{algorithm}".
\usepackage[noend]{algpseudocode}
\usepackage{amsfonts} % In order to use, for example, "\mathbb".
\usepackage{amsmath} % In order to use, for example, "\text".
\usepackage{amsthm} % In order to use, for example, "\begin{proof}".
\usepackage[english]{babel} % In order to use, for example, "\foreignlanguage".
\usepackage{color} % In order to use, for example, "\textcolor". (COMMENT THIS LINE IN THE END)
\usepackage{comment} % In order to use, for example, "\begin{comment}". (COMMENT THIS LINE IN THE END)
\usepackage[bottom=1in,left=1in,right=1in,top=1in]{geometry} % In order
\usepackage{ifthen} % In order to use, for example, "\ifthenelse".
\usepackage{mathptmx} \usepackage[T1]{fontenc} % In order to use the "Times" font required by the conference.
\usepackage{textcase} % In order to use, for example, "\MakeTextUppercase".
\usepackage{url} % In order to use, for example, "\url".
\usepackage{epsfig,color}
\usepackage{subfig}

\newcommand{\arrayfromto}
  [3] % #1: array. #2: first index. #3: last index.
  {\ensuremath{#1[#2:#3]}}

\newcommand{\algoname} % Typesets the name of an algorithm.
  [1] % #1: name of the algorithm.
  {\textsc{#1}}

\newcommand{\algoptposneg}
  {\algoname{In\-ser\-tion\-Of\-Neg\-a\-tive}}

\newcommand{\approxsorting}
  {\algoname{Ap\-prox\-Sort\-ing}}

\newcommand{\paramsorting}
  {\algoname{Pa\-ram\-e\-trized\-Sort\-ing}}
\newcommand{\arrayindex}
  [2] % #1: array. #2: index.
  {\ensuremath{#1[#2]}}

\newcommand{\assignment} % To be used inside an "algorithm2e" algorithm.
  [3] % #1 (opt): controls what is appended to #3:
  [semicolon]
  {#2 \ensuremath{\leftarrow} #3\ifthenelse{\equal{#1}{comma}}
                                           {,\ }
                                           {\ifthenelse{\equal{#1}{nosemicolon}}{}{ \;}}}

\newcommand{\beginofinterv} % The beginning of interval index #1.
  [1] % #1: number of the interval.
  {\ensuremath{\mathit{i_{#1}}}}

\newcommand{\booleanfalse} % For use in an "algorithm2e" algorithm.
  {\mbox{\KwSty{false}}}

\newcommand{\booleantrue} % For use in an "algorithm2e" algorithm.
  {\mbox{\KwSty{true}}}

\newcommand{\booleantype} % For use in an "algorithm2e" algorithm.
  {\mbox{\KwSty{bool}}}

\newcommand{\breakkw} % For use in an "algorithm2e" algorithm.
  {\KwSty{break} \;}

\newcounter{casocont}
\newcommand{\caso} % Para ser usado juntamente com o ambiente "chavedecasos".
  [3] % 1: largura do texto da condição (por padrão, o tamanho do texto); 2: o valor no caso; 3: o texto da condição.
  [nada]
  {\ifthenelse{\value{casocont}=0}
              {}
              {\\}%
   {#2},\hspace*{1ex}%
   \ifthenelse{\value{chavedecasosalin} = 0}
              {}
              {&}%
   \ifthenelse{\equal{#1}{nada}}
              {\emmodopar{#3}}
              {\emmodopar[#1]{#3}}%
   \addtocounter{casocont}{1}}

\newcounter{chavedecasosalin}
  [1]
  {\left\{\ifthenelse{\equal{#1}{0}}
                     {\begin{array}{@{}l}\setcounter{chavedecasosalin}{0}}
                     {\begin{array}{@{}l@{}l}\setcounter{chavedecasosalin}{1}}%
   \setcounter{casocont}{0}}
  {\end{array}\right.}

\newlength{\emmodoparlargura}
\newcommand{\emmodopar} % Tipografa em modo de parágrafo.
  [2] % 1: largura da mini-página (por padrão, a largura do argumento); 2: o texto a ser tipografado.
  [\emmodoparlargura]
  {\settowidth{\emmodoparlargura}{{#2}}\begin{minipage}[t]{#1}{#2}\end{minipage}}
\newcommand{\curveaccsum} % Curve (function) of accumulated sum for sequence #1.
  [1] % #1: sequence.
  {\ensuremath{M_{#1}}}

\newcommand{\curveaccsumof} % Function \curveaccsum{#1} applied to argument #2.
  [2] % #1: sequence. #2: argument.
  {\ensuremath{\curveaccsum{#1}(#2)}}

\newcommand{\delimit}
  [4] % #1: size of the delimiters: 1 - normal, * - extensible.
  {\ensuremath{\ifthenelse{\equal{#1}{*}}{\left}{}#2
               #4
               \ifthenelse{\equal{#1}{*}}{\right}{}#3}}

\newcommand{\Endofinterv} % The end of interval index #1.
  [1] % #1: number of the interval.
  {\ensuremath{\mathit{j_{#1}}}}

\newcommand{\fdcd} % Function, domain, co-domain.
  [3] % #1: name of the function. #2: domain. #3: co-domain.
  {\ensuremath{#1\colon #2 \to #3}}

\newcommand{\fdcddef} % Function, domain, co-domain and definition.
  [5] % #1--#3: the arguments of "\fdcd". #4: sample argument. #5: value of the sample argument.
  {\ensuremath{\fdcd{#1}{#2}{#3}\colon #4 \mapsto #5}}

\newcommand{\fieldname} % Typesets the name of a field (of a structure, of a class, etc).
  [1] % Field name.
  {\variablename{#1}}

\newcommand{\floor}
  [2] % #1: size of the delimiters (see #1 of \delimit).
  [1]
  {\delimit{#1}{\lfloor}{\rfloor}{#2}}

\newcommand{\insertintoseq}
  [3] % #1: new element. #2: original sequence. #3: position of the new element.
  {\seqaftinsin{#2}{#3}}

\newcommand{\interv} % Interval index #1.
  [1] % #1: index of the interval.
  {\ensuremath{I_{#1}}}

\newcommand{\newpos}
  [1] % #1: p.
  {\ensuremath{\mathit{new}_{#1}}}

\newcommand{\newposof}
  [2] % #1: p. #2: old position.
  {\ensuremath{\newpos{#1}(#2)}}

\newcommand{\objfield} % "Object.field".
  [2] % #1: object. #2: field.
  {#1\mbox{\textnormal{\textbf{.}}}\fieldname{#2}}

\newcommand{\objtype} % "Object : type".
  [2] % #1: object. #2: type.
  {#1\,\mbox{\textnormal{\textbf{:}}}\,\typename{#2}}

\newcommand{\oldpos}
  [1] % #1: p.
  {\ensuremath{\mathit{old}_{#1}}}

\newcommand{\oldposof}
  [2] % #1: p. #2: new position.
  {\ensuremath{\oldpos{#1}(#2)}}

\newcommand{\Omegaof}
  [1] % Function.
  {\ensuremath{\bigOmega(#1)}}
\newcommand{\bigOmega}{\ensuremath{\Omega}}

\newcommand{\Oof}
  [1] % Function.
  {\ensuremath{\bigO(#1)}}
\newcommand{\bigO}{\ensuremath{O}}

\newcommand{\optposinterv} % Optimal insertion position for interval #1.
  [1] % #1: index of the interval.
  {\ensuremath{\mathit{opt}(#1)}}

\newcommand{\peakininterv}
  [2] % #1: interval of insertion. #2: interval of the peak.
  {\ensuremath{f(#1,#2)}}

\newcommand{\peakof}
  [1] % #1: insertion position.
  {\ensuremath{\mathit{peak}(#1)}}

\newcommand{\insertionproblem}
  {\Problemname{ISS}}

\newcommand{\kpartitionproblem}
  {\Problemname{Multiprocessor Scheduling}}

\newcommand{\Problemname}
  [1] % #1: name of the problem.
  {\textsc{#1}}

\newcommand{\restrictedsortingproblem}
  {\Problemname{SSS$(k,s)$}}

\newcommand{\sortingproblem}
  {\Problemname{SSS}}
\newcommand{\reminder}
  [1] % The contents of the reminder.
  {\textcolor{red}{#1}}

\newcommand{\seq}
  [2] % #1: Size of the delimiters (see #1 of "\delimit").
  [1]
  {\delimit{#1}{\langle}{\rangle}{#2}}

\newcommand{\seqaftinsin} % SEQuence #1 AFTer INSertion IN position #2.
  [2] % #1: sequence. #2: position.
  {\ensuremath{#1^{(#2)}}}

\newcommand{\seqfromto}
  [3] % #1: name of the sequence. #2: first index. #3: last index.
  {\ensuremath{\seq{#1_{#2}, \ldots, #1_{#3}}}}

\newcommand{\seqindex}
  [2] % #1: sequence. #2: index.
  {\ensuremath{#1_{#2}}}

\newcommand{\set}
  [2] % #1: Size of the delimiters (see #1 of "\delimit").
  [1]
  {\delimit{#1}{\{}{\}}{#2}}

\newcommand{\sizeofseq}
  [2] % #1: Size of the delimiters (see #1 of "\delimit").
  [1]
  {\protect{\sizeofset[1]{#2}}}

\newcommand{\sizeofset}
  [2] % #1: Size of the delimiters (see #1 of "\delimit").
  [1]
  {\delimit{#1}{\lvert}{\rvert}{#2}}

\newcommand{\subseq}
  [3] % #1: original sequence. #2: beginning of the subsequence. #3: end of the subsequence.
  {\ensuremath{#1_{#2}^{#3}}}
\newcommand{\sumbest} % Variable of Kadane's algorithm (unfortunately declared globally).
  {\variablename{sum\_best}}

\newcommand{\sumcurr} % Variable of Kadane's algorithm (unfortunately declared globally).
  {\variablename{sum\_curr}}

\newcommand{\sumofseq}
  [1] % Sequence.
  {\ensuremath{\mathit{score}(#1)}}

\newcommand{\sumofset}
  [1] % #1: set.
  {\ensuremath{\mathit{score}(#1)}}

\newcommand{\tendsto} % For use in limit expressions.
  {\ensuremath{\rightarrow}}

\newcommand{\prefixsum}
  [1] % Index
  {\ensuremath{\mathit{prefixsum}(#1)}}

\newcommand{\ternaryoperator} % For use in an "algorithm2e" algorithm.
  [3] % #1: condition. #2: value in case condition is true. #3: value in case condition is false.
  {#1 \ \KwSty{?} \ #2 \ \KwSty{:} \ #3}

\newtheorem{lemma}{Lemma}

\newtheorem{theorem}{Theorem}
\newtheorem{fact}{Fact}
\newtheorem{observation}{Observation}
\newcommand{\Thetaof}
  [1] % Function.
  {\ensuremath{\bigTheta(#1)}}
\newcommand{\bigTheta}{\ensuremath{\Theta}}

\newcommand{\typename} % Typesets the name of a type in an algorithm.
  [1] % Type name.
  {\variablename{#1}}

\newcommand{\valueofseq}
  [1] % Sequence.
  {\ensuremath{score^*(#1)}}

\newcommand{\variablename} % Typesets a variable name (in roman type) in an algorithm.
  [1] % #1: variable name.
  {\ensuremath{\text{\textnormal{#1}}}}

\begin{document}
\title{Insertion and Sorting in a Sequence of Numbers Minimizing the Maximum Sum
of a Contiguous Subsequence\thanks{This work is partially supported by
FUNCAP/INRIA (Ceará State, Brazil/France) and CNPq (Brazil) research projects. A slightly different version of this paper has been 
submitted for journal publication.}}

\author
{%
Ricardo C. Corrêa, Pablo M. S. Farias\thanks{Partially supported by a doctoral
scholarship of CAPES (Programa de Demanda Social).} \\
{\small ParGO Research Group}\thanks{\texttt{http://www.lia.ufc.br/\~{}pargo}}
\newcounter{pargo}\setcounter{pargo}{\value{footnote}}
\\
{\small Universidade Federal do Ceará} \\
{\small Campus do Pici, Bloco 910} \\
{\small 60440-554 Fortaleza, CE, Brazil} \\
{\small \texttt{\{correa,pmsf\}@lia.ufc.br}} \\
   \and
Críston P. de Souza%
  \thanks{Partially supported by a FUNCAP grant.}
\\
\setcounter{footnote}{\value{pargo}}
{\small ParGO Research Group}$^\fnsymbol{pargo}$
\\
{\small Universidade Federal do Ceará} \\
{\small Campus de Quixadá} \\
{\small 63900-000 Quixadá, CE, Brazil}\\
{\small \texttt{cristonsouza@lia.ufc.br}}
}

\date{}

\maketitle

\begin{abstract}
Let $A$ be a sequence of $n \geq 0$ real numbers. A subsequence of $A$ is a sequence of contiguous elements of $A$. A 
\emph{maximum scoring subsequence} of $A$ is a subsequence with largest sum of its elements
, which can be found in $O(n)$ time by Kadane's dynamic programming algorithm. We consider in this paper two problems involving maximal scoring subsequences of a
sequence. Both of these problems arise in the context of buffer memory minimization in computer networks. 
The first one, which is called \Problemname{Insertion in a Sequence with Scores (\insertionproblem)}, consists in inserting a given real number $x$ in $A$ in such a way to minimize the sum of a maximum scoring subsequence of the resulting sequence, which can be easily done in $O(n^2)$ time by successively applying Kadane's algorithm to compute the maximum scoring subsequence of the resulting sequence corresponding to each possible insertion position for $x$. We show in this paper that the ISS problem can be solved in linear time and space with a more specialized algorithm. The second problem we consider in this paper is the \Problemname{Sorting a Sequence by Scores (\sortingproblem)} one, stated as follows: find a permutation $A'$ of $A$ that minimizes the sum of a maximum scoring subsequence. We show that the SSS problem is 
strongly NP-Hard and give a 2-approximation algorithm for it.
\end{abstract}

\section{Introduction}

Let the elements of a sequence $A$ of $n \geq 0$
real numbers be denoted by \seqindex{a}{1}, \seqindex{a}{2},
\ldots, \seqindex{a}{n}. Then, $A$ is the sequence \seq{\seqindex{a}{1},
\seqindex{a}{2}, \ldots, \seqindex{a}{n}} (which is \seq{} if $n=0$) and its
size is $\sizeofseq{A} = n$. A subsequence of $A$ defined by indices $0 \leq i
\leq j \leq n$ is denoted by \subseq{A}{i}{j}, which equals either \seq{}, if
$i = j$, or the sequence \seq{\seqindex{a}{i+1}, \ldots, \seqindex{a}{j}} of
contiguous elements of $A$, otherwise (see Figure~\ref{fig:sequence} for an
example).
Let $\sumofseq{\subseq{A}{i}{j}} = \sum_{k=i+1}^{j} \seqindex{a}{k}$ stand for the
sum of elements of $\subseq{A}{i}{j}$ (we consider
$\sumofseq{\seq{}}=0$). A \emph{maximum scoring subsequence} of $A$ is a
subsequence with largest score. The \Problemname{Maximum Scoring Subsequence
(MSS)} problem is that of finding a maximum scoring subsequence of a given
sequence $A$. The MSS problem can be solved in $O(n)$ time
by Kadane's dynamic programming algorithm, whose essence
is to consider $A$ as a concatenation $\seq{\subseq{A}{0}{j_1}, \subseq{A}{i_2 =
j_1}{j_2}, \ldots, \subseq{A}{i_\ell}{j_\ell}}$ of appropriate subsequences,
called {\em intervals}, and to determine $S_k$ as a maximum scoring subsequence
of $\subseq{A}{i_k}{j_k}$, for all $k \in \{ 1, 2, \ldots, \ell \}$. Defining
each interval $\subseq{A}{i_k}{j_k}$ -- with the possible exception of the last
one -- to be such that $\sumofseq{\subseq{A}{i_k}{j_k}} < 0$ and
$\sumofseq{\subseq{A}{i_k}{j'}} \geq 0$, for all $i_k \leq j' < j_k$, then the
largest score subsequence among $\set{S_1, S_2, \ldots, S_\ell}$ is a maximum scoring subsequence of $A$ \cite{Bentley.84,Gries.82}. The \emph{value} of $A$
is $\valueofseq{A} = \sumofseq{S}$, for any maximum scoring subsequence $S$ of
$A$.

\begin{figure}[htb]
\centerline {\input{sequence.pstex_t}}
\caption{An example of a sequence and a subsequence. A maximum scoring
subsequence is $A_{13}^{18}$ and $\valueofseq{A} = 12$.}
\label{fig:sequence}
\end{figure}

The MSS problem has several applications in practice, where maximum scoring
subsequences correspond to various structures of interest. For instance, in
Computational Biology, in the context of certain amino acid scoring schemes and
several other applications mentioned in \cite{Csuros.2004,Ruzzo.Tompa.99}.
In such a context, it may also be useful to find not only one but a maximal set
of non-overlapping maximum scoring subsequences of a given sequence $A$. This can
be formalized as the ALL MAXIMAL SCORING SUBSEQUENCES problem, for which have
been devised a linear sequential algorithm \cite{Ruzzo.Tompa.99}, a PRAM EREW
work-optimal algorithm that runs in \Oof{\log n} time and makes \Oof{n}
operations \cite{Dai.Su.2006} and a BSP/CGM parallel algorithm which uses $p$
processors and takes \Oof{\sizeofseq{A}/p} time and space per processor
\cite{Alves.Caceres.Song.06}. The MSS problem has also been generalized in the
direction of finding a list of $k$ (possibly overlapping) maximum scoring
subsequences of a given sequence $A$. This is known as the $k$ MAXIMUM SUMS
PROBLEM \cite{Bae.Takaoka.04} and for a generalization of it an optimal
\Oof{n+k} time and \Oof{k} space algorithm has been devised
\cite{Brodal.Jorgensen.2007,Brodal.Jorgensen.2008}. An optimal \Oof{n
\cdot \max \{1, \log (k/n)\}} algorithm has also been developed for the related
problem of selecting the subsequence with the $k$-th largest score
\cite{Brodal.Jorgensen.2008}. 

We consider in this paper two problems related to the MSS.
The first one, which is called \Problemname{Insertion in a Sequence with scores (\insertionproblem)},
consists in inserting a given real number $x$ in $A$ in such a way to minimize
the maximum score of a subsequence of the resulting sequence. The operation of
{\em inserting $x$ in $A$} is associated with an {\em insertion index} $p \in
\set{0, \ldots, n}$ and the {\em resulting sequence} $\insertintoseq{x}{A}{p} =
\seq{\subseq{A}{0}{p}, x, \subseq{A}{p}{n}}$, that is, the sequence obtained by the concatenation of
$\subseq{A}{0}{p}$, $x$, and $\subseq{A}{p}{n}$. The objective of the
\insertionproblem\ problem is to determine an insertion index $p^*$ that
minimizes $\valueofseq{\insertintoseq{x}{A}{p^*}}$, which can
be easily done in $O(n^2)$ time and $O(n)$ space by successively using Kadane's algorithm to compute the maximum scoring
subsequence of \insertintoseq{x}{A}{0}, \ldots, \insertintoseq{x}{A}{n}. We show
in this paper that we can do better. More precisely, we show that the ISS
problem can be solved in linear time.

The {\insertionproblem} problem can be approached more specifically
depending on the value of $x$. The case $x=0$ is trivial since
$\valueofseq{\insertintoseq{x}{A}{p}} = \valueofseq{A}$ independently of the
value of $p$, which means that any insertion index $p$ is optimal for $A$. If $x
< 0$, then $\sumofseq{\insertintoseq{x}{A}{p}} < \sumofseq{A}$, for all
insertion indices $p \in \set{0, 1, \ldots, n}$. Intuitively, then, $x$ has to
be inserted inside some maximum scoring subsequence $S = \subseq{A}{i}{j}$ of
$A$, in an attempt to reduce the value of \insertintoseq{x}{\linebreak[0]A}{\linebreak[0]p} with respect to that of $A$.
Even though the value of $\insertintoseq{x}{A}{p}$ cannot be smaller than
$\valueofseq{A}$ in certain cases (for instance, if $S$ has only one positive
element, or $\valueofseq{\subseq{A}{0}{i}} = \sumofseq{S}$, or
$\valueofseq{\subseq{A}{j}{n}} = \sumofseq{S}$, then all insertion indices are
equally good for $A$ since $\valueofseq{\insertintoseq{x}{A}{p}} =
\valueofseq{A}$ for any particular choice of $p$), we describe an $O(n)$
time and space algorithm to determine a best insertion position in a
maximum scoring subsequence of $A$, provided that $x$ is negative.

Showing that the ISS problem can be solved in linear time is a more complex
task when $x > 0$. Inserting $x$ inside a maximum scoring subsequence $S$
of $A$ will certainly lead to a subsequence $S'$ of \insertintoseq{x}{A}{p} such
that $\sumofseq{S'} > \sumofseq{S}$ (this may happen even if $x$ in inserted
outside $S$). Intuitively, therefore, we should choose an insertion position
where $x$ can only ``contribute'' to subsequences whose scores are as small as possible.
Computing the necessary information for this in \Oof{n} time may seem hard at
first, but we can make things simpler by considering the partition into
intervals of $A$ (the same used in Kadane's algorithm). The idea is to
determine the interval $\subseq{A}{i_k}{j_k}$ having an optimal insertion index. The difficulty to
accomplish this task in linear time stems from the fact that computing
$\valueofseq{\insertintoseq{x}{A}{p}}$ when $p$ is an insertion index in an
interval $\subseq{A}{i_k}{j_k}$ may involve one or more intervals other than
$\subseq{A}{i_k}{j_k}$. We overcome this difficulty by means of a dynamic
programming approach.

The second problem we consider in this paper is the \Problemname{Sorting a
 sequence by scores (\sortingproblem)}, stated as follows: given the sequence $A$, find a
 permutation $A'$ of $A$ that minimizes \valueofseq{A'}. The \sortingproblem\ problem is referred to as
the \Problemname{Sequencing to Minimize the Maximum Renewal Cumulative Cost} in~\cite{Tsai.92}. Among other applications, this latter problem models buffer memory
usage in a node of a computer network. In this case, the absolute value of a
number models the local memory space required to store a corresponding message after its
reception and before it is resent through the network (in practice, there are additional cases in which the message is produced or consumed 
locally; these situations are ignored in this high level description for the sake of simplicity of exposition). 
This behavior can be seen as the execution of tasks (sending or receiving messages), each of which is associated with a (positive or negative)
cost that corresponds to the additional units of resources (local memory space) that are occupied after its execution.
Receiving a message results in a positive cost, while sending a message can be viewed as effecting a negative cost.
In this context, finding maximum scoring subsequences of sequences defining
communications between the nodes of a network corresponds to finding the
greatest buffer usage in each node \cite{Vieira.Rezende.Barbosa.Fdida.2012}.
Moreover, when the intention is to find an ordering for these communications
with the aim of minimizing the resulting memory usage, then we are left with the problem of sorting the communications so as to
minimize the maximum renewal cumulative cost.

It is mentioned in~\cite{Tsai.92} that the \sortingproblem\ problem has been proved to be strongly NP-hard by means of a
transformation from the \Problemname{3-Partition} problem. Indeed, a straightforward reduction from \Problemname{3-Partition} yields that the \sortingproblem\ problem remains 
NP-hard in the strong sense even if all negative elements in $A$ are equal to a value $-s$ and every positive 
element is in a certain range depending on $s$ (more details are given in Section~\ref{sec:sort}).
It is known that the \sortingproblem\ problem becomes polynomially solvable if the negative elements are $-s$ and the positive elements
are all equal to some value $s'$~\cite{Tsai.92}. In this paper, we devise a $(1 + M/\valueofseq{A})$-approximation 
algorithm for the \sortingproblem\ problem, where $M$ is the maximum element in $A$, 
which runs in \Oof{n \log n} time. For the general case of the \sortingproblem\ problem, since $\valueofseq{A} \geq M$, this algorithm has approximation 
factor of 2, and we show that this factor is tight.
However, for the aforementioned more particular case where the elements of $A$ are bounded, the approximation factor of this same 
algorithm becomes $3/2$, for $n \geq 3$.

We organize the remainder of the text as follows. Section~\ref{sec:prelim} states some
useful properties of maximum score subsequences for later use. In
Section~\ref{section-negative-case} and Section~\ref{section-positive-case} we then present our
solutions to the ISS problem for the cases where the inserted number $x$ is
negative and positive, respectively. Section~\ref{sec:sort} contains our results on
the SSS problem, and Section~\ref{sec:conc} finally provides conclusions and
directions for further investigations.

\section{Preliminaries on the ISS problem}
\label{sec:prelim}

Let us establish some simple and useful properties of sequence $A$ and a
subsequence $\subseq{A}{i}{j}$, for $0 \leq i \leq j \leq n$. We start with
three properties that give a view of minimal (with respect to inclusion) maximum
scoring subsequences. Let a {\em prefix} ({\em suffix}) of $A_i^j$ be a
subsequence $\subseq{A}{i}{j'}$ ($\subseq{A}{i'}{j}$), with $i \leq j' \leq j$ ($i \leq i' \leq j$).

\begin{fact}
If $A_i^j$ is a maximum scoring subsequence of $A$, then its prefixes and
suffixes have all nonnegative scores, otherwise a larger scoring subsequence can be
obtained by deleting a prefix or a suffix of negative score. Conversely,
$\sumofseq{X} \leq 0$, where $X$ is any suffix of $\subseq{A}{0}{i}$ or prefix
of $\subseq{A}{j}{n}$, otherwise a larger scoring subsequence
can be obtained by
concatenating $A_i^j$ with a suffix of $\subseq{A}{0}{i}$ or prefix of
$\subseq{A}{j}{n}$ of positive score.
\label{fact:nonnegprefsuf}
\end{fact}

\begin{fact}
If $\subseq{A}{i}{j}$ is a maximum scoring subsequence of
$A$, then there is a maximum scoring subsequence of $\subseq{A}{i}{j}$ which is
a prefix (suffix) of $\subseq{A}{i}{j}$.
\label{fact:prefmaxpref}
\end{fact}

\begin{figure}[htb]
\centerline {\input{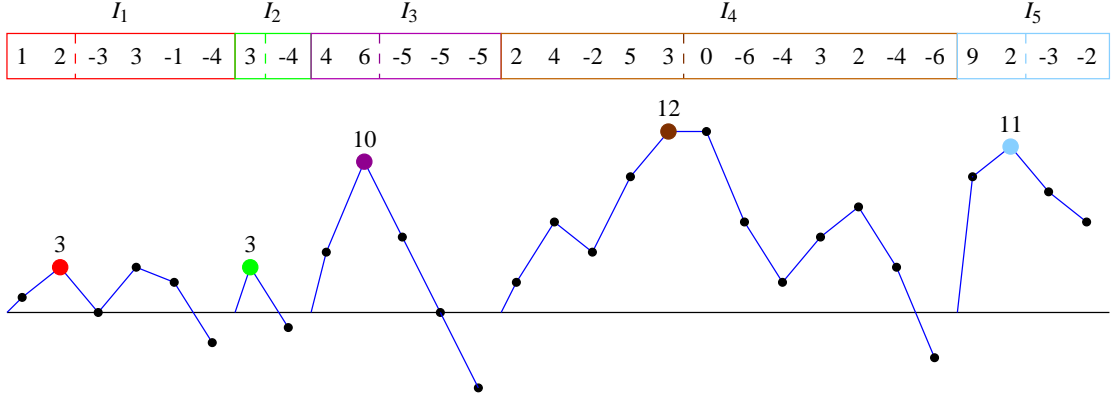}}
\caption{Partition into intervals of the sequence in
Figure~\ref{fig:sequence}. For each interval, the score of its prefixes is
indicated, as well as its maximum scoring subsequence.}
\label{fig:intervals}
\end{figure}

The definitions in the sequel are illustrated in Figure~\ref{fig:intervals}. The
subsequence $\subseq{A}{i}{j}$ is an {\em interval} if $\sumofseq{\subseq{A}{i}{j}} < 0$
or $j = n$, and $\sumofseq{\subseq{A}{i}{j'}} \geq 0$, for all $i \leq j' < j$.
The {\em partition into intervals} of $A$ is the concatenation $\seq{I_1 =
\subseq{A}{0}{j_1}, I_2 = \subseq{A}{i_2 = j_1}{j_2}, \ldots, I_\ell =
\subseq{A}{i_\ell}{j_\ell}}$ of the $\ell$ maximal intervals of $A$. Such a
partition is explored in Kadane's algorithm due to the fact that a maximum
scoring subsequence of $A$ is a subsequence of some of its intervals.

\begin{fact}
If $A_i^j$ is a maximum scoring subsequence of interval $I_k$ and $A_{i'}^{j'}$
is a prefix (suffix) of $A_i^j$ such that $\sumofseq{A_{i'}^{j'}} = 0$, then
$A_i^j \setminus A_{i'}^{j'}$ is a maximum scoring subsequence of $I_k$.
\label{fact:maxispref}
\end{fact}

While the previous properties are general for every sequence, the next one
is more specific to the resulting sequence of an insertion. Recall that $x$
stands for the real number given as input to the ISS problem.
Assume that the insertion index $p$ is such that $i_k \leq p < j_k$, which
means that $x$ is inserted in $I_k$.

\begin{fact}
The score of all elements of $I_k$ whose indices are greater than $p$
are affected by the insertion of $x$ in the following way: for every $p < q
\leq j_k+1$, $\sumofseq{\subseq{{\insertintoseq{x}{A}{p}}}{i_k}{q}} =
\sumofseq{\subseq{A}{i_k}{q-1}} + x$.
\label{fact:inssumx}
\end{fact}

This fact is the reason why the discussion of cases $x < 0$ and $x >
0$ is carried out separately in the two next sections.
For the positive case, since all prefixes of $I_k$ have nonnegative scores
(Fact~\ref{fact:nonnegprefsuf}), consecutive intervals may be merged in the
resulting sequence, provided that $x$ is large enough to
make $\sumofseq{\subseq{{\insertintoseq{x}{A}{p}}}{i_k}{j_k+1}} > 0$.
For instance, consider interval $I_1$ in Figure~\ref{fig:intervals}. The
insertion of $x = 6$ at the very end of this interval (i.e, at insertion
position $p = j_1-1 = 5$) creates the subsequence $\seq{\subseq{A}{i_1
= 0}{5}, 6, -4}$ and the new interval $\seq{\subseq{A}{0}{5}, 6,
-4, I_2, I_3}$. On the other hand, for the
negative case, the insertion of $x$ may split $I_k$ into two or more intervals
if there exists $p \leq q \leq j_k$ such that
$\sumofseq{\subseq{{\insertintoseq{x}{A}{p}}}{i_k}{q}} < 0$, in which case
$\subseq{{\insertintoseq{x}{A}{p}}}{i_k}{q}$ is an interval of
$\insertintoseq{x}{A}{p}$ but $\subseq{{\insertintoseq{x}{A}{p}}}{i_k}{j_k}$ is
not. Again in Figure~\ref{fig:intervals}, the insertion of $x = -6$ between the elements -2 and 5 of interval $I_4$ splits it into 3 intervals, namely $\seq{2, 4, -2, -6}$, $\seq{5, 3, 0, -6, -4}$, and $\seq{3, 2, -4, -6}$.

\section{Inserting $x < 0$} \label{section-negative-case}

As already mentioned in the Introduction, solving the ISS problem when $x < 0$ corresponds to
insert $x$ in some maximum scoring subsequence $\subseq{A}{i}{j}$. According to
Fact~\ref{fact:maxispref}, we assume that $\subseq{A}{i}{j}$ is minimal with respect to
inclusion. What remains to be specified is the way to find an appropriate
insertion index in $\subseq{A}{i}{j}$. The cases $n = 0$, $j \leq i + 1$, and $\sumofseq{\subseq{A}{i}{j}} = 0$ are trivial. Then, assume that $n > 0$, $j > i + 1$, and
$\sumofseq{\subseq{A}{i}{j}} > 0$. Inserting $x$ inside $\subseq{A}{i}{j}$ divides the latter in its left (a prefix
of $\subseq{A}{i}{j}$) and right (a suffix of $\subseq{A}{i}{j}$) parts, and
different choices of $p$ may lead to different values of
\insertintoseq{x}{A}{p}, as depicted in Figure~\ref{fig:neginss}. Using Fact~\ref{fact:prefmaxpref}, the algorithm computes
the insertion index $i < p < j$ such that the maximum between
$\valueofseq{\subseq{A}{i}{p}}$
and $\valueofseq{\subseq{A}{p}{j}}$ is as small as possible.
Such a computation can be carried out by simply performing a left-to-right traversal of $\subseq{A}{i}{j}$ to compute (and store) 
the values of all possible prefixes of $\subseq{A}{i}{j}$, and a further right-to-left traversal
to compute the values of all possible suffixes of
$\subseq{A}{i}{j}$. This strategy is materialized in
Algorithm {\algoptposneg}, which receives as input an array with the elements of $A$ and the number $x < 0$, and returns $p$ 
computed as above.

\begin{figure}[htb]
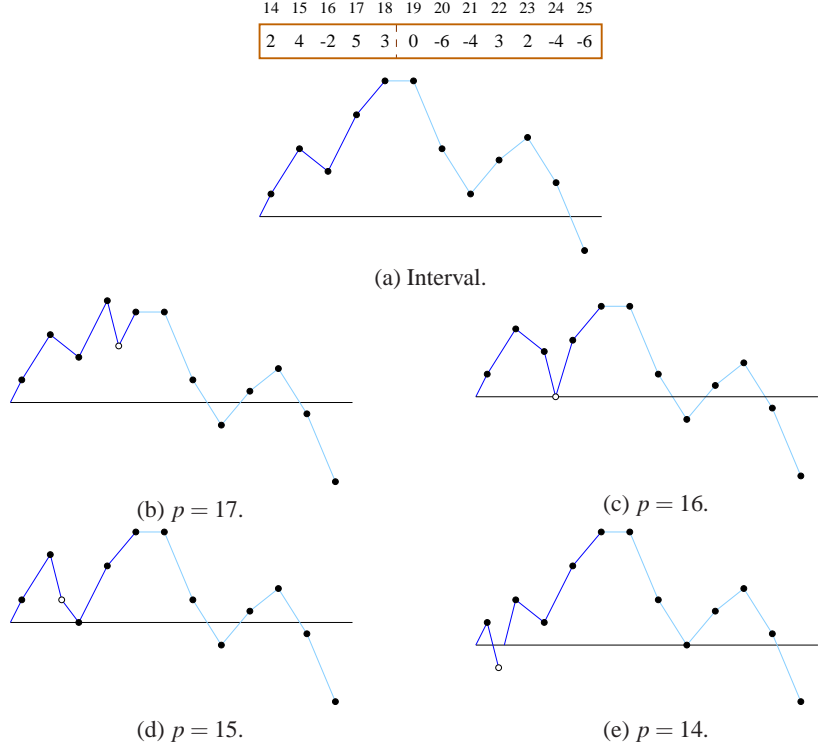

\centering

\subfloat[position=bottom][Interval.]{
\begin{minipage}{\textwidth}
\centerline{\input{negins.pstex_t}}
\end{minipage} \label{sfig:negins}}
\quad\quad\quad
\subfloat[position=bottom][$p = 17$.]{
\begin{minipage}{0.30\textwidth}
\input{negins1.pstex_t}
\end{minipage} \label{sfig:negins1}}
\quad\quad\quad
\subfloat[position=bottom][$p = 16$.]{
\begin{minipage}{0.30\textwidth}
\input{negins2.pstex_t}
\end{minipage} \label{sfig:negins2}}
\quad\quad\quad
\subfloat[position=bottom][$p = 15$.]{
\begin{minipage}{0.30\textwidth}
\input{negins3.pstex_t}
\end{minipage} \label{sfig:negins3}}
\quad\quad\quad
\subfloat[position=bottom][$p = 14$.]{
\begin{minipage}{0.30\textwidth}
\input{negins4.pstex_t}
\end{minipage} \label{sfig:negins4}}

\caption{Possible insertion positions in the interval $I_4$ of the example in
Figure~\ref{fig:intervals} for $x = -4$.}
\label{fig:neginss}
\end{figure}

\begin{lemma}
 Algorithm {\algoptposneg}$(A, x)$
 returns an optimal insertion index, provided that $x < 0$. In addition, it
 runs in $O(n)$ time and space.
 \label{thm:insneg}
\end{lemma}

\begin{proof}
Let $p \in \set{i+1, \ldots, j-1}$ be the value computed by the algorithm and
$p' \neq p$ be another arbitrary insertion index. We show that
$\valueofseq{\insertintoseq{x}{A}{p}} \leq \valueofseq{\insertintoseq{x}{A}{p'}}$. Let in addition $T$ be a maximum scoring subsequence of
$\insertintoseq{x}{A}{p}$, minimal with respect to inclusion. Note that $T \ne \seq{}$ since $\valueofseq{A} =
\sumofseq{\subseq{A}{i}{j}} > 0$. Moreover, by
Fact~\ref{fact:nonnegprefsuf}, $x$ is neither the first nor the last element
of $T$. So, let $y$ and $z$ be such that $\subseq{T}{0}{1} = \seq{a_{y+1}}$ and
$\subseq{T}{|T|-1}{|T|} = \seq{a_z}$. The first case to be analyzed is when $x$ is in $T$, i.e. $y < p < z$ (Figure~\ref{fig:thmneg}\subref{sfig:xinT}).
In this case, by Fact~\ref{fact:nonnegprefsuf} and the minimality of
$\subseq{A}{i}{j}$ and $T$, $y = i$ and $z = j$ or, in other words, $T =
\seq{\subseq{A}{i}{p}, x, \subseq{A}{p}{j}}$. The elements of $\subseq{A}{i}{j}$
also form, perhaps with the occurrence of $x$ at some position, a subsequence
$T'$ of $\insertintoseq{x}{A}{p'}$, and since $x < 0$, we conclude that $\sumofseq{T'} \geq \sumofseq{T}$ (equality holds if $y < p' < z$).
Then $\valueofseq{\insertintoseq{x}{A}{p}} = \sumofseq{T} \leq \sumofseq{T'}
\leq \valueofseq{\insertintoseq{x}{A}{p'}}$, as claimed.

Assume that $p \notin \set{y, \ldots, z}$. If $T$'s elements also
form a subsequence of $\insertintoseq{x}{A}{p'}$ (more precisely, $p' \notin
\set{y+1, \ldots, z-1}$), then $\valueofseq{\insertintoseq{x}{A}{p}} =
\sumofseq{T} \leq \valueofseq{\insertintoseq{x}{A}{p'}}$, as desired. Then, assume that $p'
\in \set{y+1, \ldots, z-1}$. If $\subseq{A}{i}{j}$ and $T$ are disjoint, then
\subseq{A}{i}{j} is also a subsequence of $\insertintoseq{x}{A}{p'}$. It turns out that $\valueofseq{\insertintoseq{x}{A}{p'}} \leq
\valueofseq{A} = \sumofseq{\subseq{A}{i}{j}}$ yields $\valueofseq{\insertintoseq{x}{A}{p'}} =
\sumofseq{\subseq{A}{i}{j}} = \valueofseq{A} \geq \valueofseq{\insertintoseq{x}{A}{p}}$.

Finally, we are left with the case when $\subseq{A}{i}{j}$ and $T$ are not
disjoint (Figure~\ref{fig:thmneg}\subref{sfig:SintT}).
By Fact~\ref{fact:nonnegprefsuf} and the minimality of $T$, either $y = i$ or $z
= j$.
Without loss of generality, let us suppose the first equality, since the other one is analogous.
We have that $\max
\set{\valueofseq{\subseq{A}{i}{p'}}, \valueofseq{\subseq{A}{p'}{j}}} \geq 
\max \set{\valueofseq{\subseq{A}{i}{p}}, \valueofseq{\subseq{A}{p}{j}}} \geq \valueofseq{\subseq{A}{i}{p}} =
\sumofseq{T}$.
The result follows since both $\subseq{A}{i}{p'}$ and $\subseq{A}{p'}{j}$ are
subsequences of $\insertintoseq{x}{A}{p'}$.

The complexities stem directly from the facts that the algorithm employs one additional
array of size $O(n)$ (for the left-to-right traversal of $A$) and performs, in addition to a call to a version of Kadane's
algorithm as a sub-routine returning the indices $i$ and $j$ and the score of
the minimal maximum scoring subsequence considered, two disjoint $O(n)$-time loops. %\qed
\end{proof}

\begin{figure}[htb]
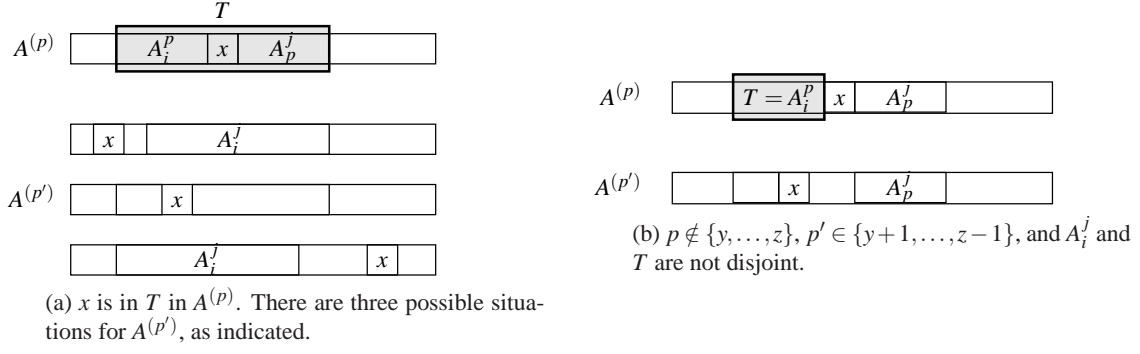

\centering

\subfloat[position=bottom][$x$ is in $T$ in
$\insertintoseq{x}{A}{p}$. There are three possible
situations for $\insertintoseq{x}{A}{p'}$, as
indicated.]{
\begin{minipage}{0.4\textwidth}
\input{neg1case.pstex_t}
\end{minipage} \label{sfig:xinT}}
\quad\quad\quad
\subfloat[position=bottom][$p \notin \set{y, \ldots, z}$, $p'
\in \set{y+1, \ldots, z-1}$, and $\subseq{A}{i}{j}$ and $T$ are not
disjoint.]{
\begin{minipage}{0.4\textwidth}
\input{neglastcase.pstex_t}
\end{minipage} \label{sfig:SintT}}

\caption{Cases of proof of Lemma~\ref{thm:insneg}.}
\label{fig:thmneg}
\end{figure}

\section{Inserting $x > 0$} \label{section-positive-case}

\newcommand{\fieldended}{\fieldname{ended}}
\newcommand{\fieldinspos}{\fieldname{pos}}
\newcommand{\fieldinterv}{\fieldname{interv}}
\newcommand{\fieldpeak}{\fieldname{peak}}
\newcommand{\vartopcurve}{\variablename{top}}

The discussion in this section is based on the partition into intervals
$\seq{I_1, I_2, \linebreak[0] \ldots, \linebreak[0] I_\ell}$ of $A$. For the sake of convenience, we
assume that $a_n = 0$ (observe that this can be done without loss of
generality since appending a new null element to $A$ does not alter the
scores of the suffixes of $A$), which means that 
$\sumofseq{I_\ell} \geq 0$. A particularity of this positive case, which is
derived from Fact~\ref{fact:inssumx}, is the following: for every interval $I_k$, index $j_k-1$ is at least as good as any other insertion index in this
interval. Thus, an optimal insertion index exists among $\Endofinterv{1}-1,
\Endofinterv{2}-1, \ldots, \Endofinterv{\ell}-1$, corresponding each one of
these indices to one interval of the partition into intervals of $A$. If
$p = \Endofinterv{k}-1$ is chosen as the insertion index, then the resulting
interval in $\insertintoseq{x	}{A}{p}$ (which may correspond to a merge of
several contiguous intervals of $A$ in the sense of Fact~\ref{fact:inssumx}) is referred
as to an {\em extended interval, relative to $I_k$} and denoted by
$\insertintoseq{x}{I}{k}$. If $I_{k'}$ is one of the intervals which are merged to produce $\insertintoseq{x}{I}{k}$, then
  $I_{k'}$ is a {\em subinterval} of $\insertintoseq{x}{I}{k}$. In the remaining
  of this section, we show a linear time algorithm to compute
$\valueofseq{\insertintoseq{x}{I}{k}}$, for all $k \in \set{1, 2, \ldots,
\ell}$. Clearly, the smallest of these values is associated with the optimal
insertion index for $x$.

For each $k$,
computing $\valueofseq{\insertintoseq{x}{I}{k}}$ by means of Kadane's
algorithm takes \Thetaof{n} time. Therefore, the exhaustive search takes
quadratic time in the worst case. However, as depicted in Figure~\ref{fig:subsums}, by graphically aligning the scores of the prefixes of
the extended intervals with respect to the intervals of $A$, one can visualize
some useful observations in connection with these curves which are explored in the
algorithm described in the sequel.
Let the sequence of negative elements composed by intervals' scores be denoted by $N = \seq{score(I_1), score(I_2),
\ldots, \linebreak[0] score(I_\ell)}$.

\begin{observation}
\label{observation-behaviour-of-a-curve}  Let $a \in I_{k'}$ be the
  element of indices $j$ in $\insertintoseq{x}{I}{k}$ and $j'$ in $I_{k'}$, $k'
  \geq k + 1$ (an assumption that is tacitly made here is that $I_{k'}$ is a
  subinterval of $\insertintoseq{x}{I}{k}$).
  Then,
\[
\begin{array}{lcl}
\sumofseq{\subseq{{\insertintoseq{x}{I}{k}}}{0}{j}} & = &
\sumofseq{\subseq{A}{i_k}{j_k-1}} + x + a_{j_k} +
\sumofseq{\subseq{A}{j_k}{i_{k'} + j'}} \\
                                                & = & x + 
\sumofseq{\subseq{A}{i_k}{j_{k'-1}}} + \sumofseq{\subseq{(I_{k'})}{0}{j'}}
\\
                                                & = & x + \sumofseq{\subseq{N}{k-1}{k'-1}} +
\sumofseq{\subseq{(I_{k'})}{0}{j'}}
\end{array}
\]
As an example, take $a = 4$, $\insertintoseq{x}{I}{k} =
\insertintoseq{x}{I}{1}$, and $I_{k'} = I_4$ in Figure~\ref{fig:subsums}. The
equality above indicates the distance of 1 between the
curves of $\insertintoseq{x}{I}{1}$ and $I_{k'}$ for the element $4 \in I_4$.
\end{observation}

A first consequence of Observation~\ref{observation-behaviour-of-a-curve} is a
recurrence relation which is used to govern our dynamic programming algorithm.
If $k < \ell$, let $\insertintoseq{x}{I}{k} \cap \insertintoseq{x}{I}{k+1}$
stand for the concatenation of the common subintervals of
$\insertintoseq{x}{I}{k}$ and $\insertintoseq{x}{I}{k+1}$ (for the sake of illustration, $\insertintoseq{x}{I}{1} \cap \insertintoseq{x}{I}{2} = \seq{I_2, I_3, I_4}$ in the example of Figure~\ref{fig:subsums}). In addition, write $I_{k'} \subseteq \insertintoseq{x}{I}{k} \cap \insertintoseq{x}{I}{k+1}$ to say that interval
$I_{k'}$ is a common subinterval of $\insertintoseq{x}{I}{k}$
and $\insertintoseq{x}{I}{k+1}$.
The recurrence for $\valueofseq{\insertintoseq{x}{I}{k}}$ is given by
\begin{equation}
\valueofseq{\insertintoseq{x}{I}{k}} = \max
\set{\valueofseq{I_k}, x + \sumofseq{\subseq{A}{i_k}{j_k-1}}},
\label{eq:baserecur}
\end{equation}
if $k = \ell$ (considering that the last element of $A$ is null) or
($k < \ell$ and $\insertintoseq{x}{I}{k} \cap \insertintoseq{x}{I}{k+1} =
\emptyset$) or, otherwise,
\begin{equation}
\max \set{\valueofseq{I_k}, x + \sumofseq{\subseq{A}{i_k}{j_k-1}}, x +
\max_{I_{k'}
\subseteq
\insertintoseq{x}{I}{k} \cap \insertintoseq{x}{I}{k+1}}
\set{\sumofseq{\subseq{N}{k-1}{k'-1}} + \valueofseq{I_{k'}}}}.
\label{eq:recur}
\end{equation}
The first two terms in~\eqref{eq:baserecur} and~\eqref{eq:recur} indicate the best
insertion index in $I_k$, while the third one in~\eqref{eq:recur} gives the best
interval in $\insertintoseq{x}{I}{k} \cap \insertintoseq{x}{I}{k+1}$ (if any).
The crucial point is then the computation of $\max_{I_{k'} \subseteq
\insertintoseq{x}{I}{k} \cap \insertintoseq{x}{I}{k+1}}
\set{\sumofseq{\linebreak[0] \subseq{N}{k-1}{k'-1}} + \valueofseq{I_{k'}}}$ when $I_{k+1}$ is
a subinterval of $\insertintoseq{x}{I}{k}$ (i.e. $\insertintoseq{x}{I}{k} \cap
\insertintoseq{x}{I}{k+1} \ne \emptyset$), which is
performed in the light of the following additional observations.

\begin{observation} \label{observation-comparing-two-curves} Let $a \in
\insertintoseq{x}{I}{k'}$ be the element of indices $j$ and $j'$ in, respectively, $\insertintoseq{x}{I}{k}$ and
$\insertintoseq{x}{I}{k'}$, $k' \geq k+1$. Write $I_{k''}$ for the
interval containing $a$, and $j''$ for the index of $a$ in $I_{k''}$.
Assuming that $k'' \ne k'$, then
\[
\begin{array}{lcl}
\sumofseq{\subseq{{\insertintoseq{x}{I}{k}}}{0}{j}}
 - \sumofseq{\subseq{{\insertintoseq{x}{I}{k'}}}{0}{j'}} & = &
 \sumofseq{\subseq{N}{k-1}{k''-1}} + \sumofseq{\subseq{(I_{k''})}{0}{j''}} -
 \sumofseq{\subseq{N}{k'-1}{k''-1}} - \sumofseq{\subseq{(I_{k''})}{0}{j''}}\\
& = & \sumofseq{\subseq{N}{k-1}{k'-1}}
\end{array}
\]
Thus, the respective curves of $\insertintoseq{x}{I}{k}$ and
$\insertintoseq{x}{I}{k'}$ remain at a constant
distance for all intervals $I_{k'} \subseteq
\insertintoseq{x}{I}{k} \cap \insertintoseq{x}{I}{k+1}$, $k' \ne k+1$, with the
curve of $\insertintoseq{x}{I}{k'}$ above that of $\insertintoseq{x}{I}{k}$.
\end{observation}

The last observation before going into the details of the algorithm is useful
to decide whether a given interval $I_{k'}$ is a subinterval of $\insertintoseq{x}{I}{k}$.

\begin{observation}
Observation~\ref{observation-behaviour-of-a-curve} implies that if interval
$I_{k'}$, $k' \geq k+1$, is contained in $\insertintoseq{x}{I}{k}$, then $x + \sumofseq{\subseq{N}{k-1}{k'-1}} \geq 0$.
  The converse is also true since $x + \sumofseq{\subseq{N}{k-1}{k'-1}} \geq 0$
  yields $x + \sumofseq{\subseq{N}{k-1}{k''-1}} \geq 0$, for all $k < k'' <
  k'$, because all members of $N$ are negative.
\end{observation}
 
\begin{figure}[htb]
\centerline {\input{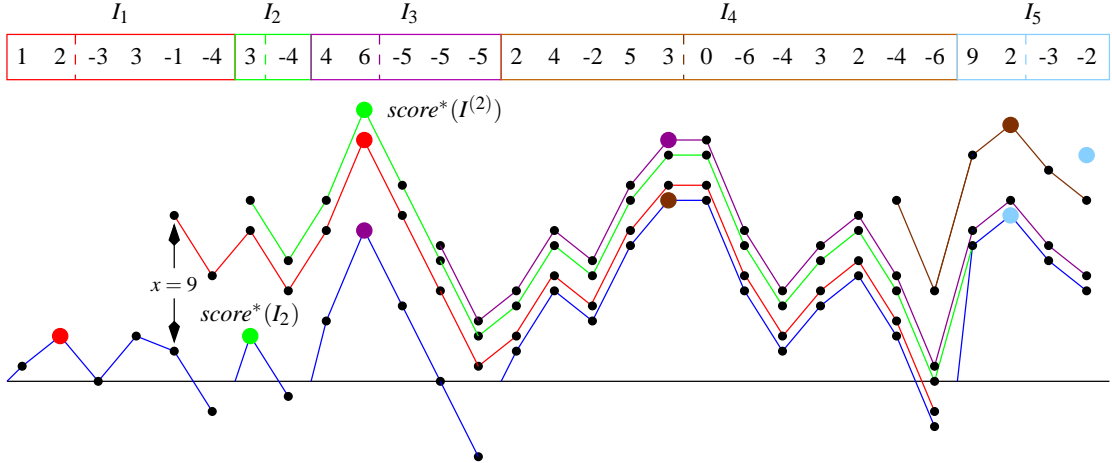}}
\caption{Scores of prefixes of all possible extended intervals resulting from
the insertion of $x = 9$ in the sequence in
Figure~\ref{fig:sequence}.
For each interval $I_k$, the points corresponding to $\valueofseq{I_k}$ and
$\valueofseq{\insertintoseq{x}{I}{k}}$ are highlighted. The last null element of
the sequence is omitted.}
\label{fig:subsums}
\end{figure}

The computation
of the largest scores of prefixes of extended intervals
$\insertintoseq{x}{I}{k}$ is divided into two phases.
The first phase is a modification of the Kadane's algorithm and its role is twofold. First, it determines the largest scores
of prefixes of $I_1, I_2, \ldots, I_\ell$ and, then, it sets the initial values of the
arrays that are used in the second phase. Such arrays are the following:
\begin{description}
\item[$SN$] suffix sums of $N$, i.e. $SN[k]$ equals
$\sumofseq{\subseq{N}{k-1}{\ell}}$, for all $k \in \set{1, 2, \ldots, \ell}$.
By definition, $\sumofseq{\subseq{N}{k-1}{k'-1}} = SN[k] - SN[k']$, for all $k'
\geq k$.
\item[$INTSCR$] largest intervals' scores, i.e. $INTSCR[k] =
\valueofseq{I_k}$, for all $k \in \set{1, 2, \ldots, \ell}$.
\item[$XSCR$] for each interval $k \in \set{1, 2, \ldots, \ell}$, this
array stores the score of the subsequence ending at $x$, provided that
$x$ is inserted in $I_k$, i.e. $XSCR[k] = x +
\sumofseq{\subseq{A}{i_k}{j_k-1}}$.
\end{description}

The second phase is devoted to the computation of the extended interval
containing the best insertion position for $x$.
This is done iteratively from $k = 1$ until $k = \ell$. For each $k$, the
recurrence relation~\eqref{eq:baserecur}--\eqref{eq:recur} is used to start the
computation of $\valueofseq{\insertintoseq{x}{I}{k}}$ and to update the maximum
score of extended intervals started in previous iterations as described in
Algorithm~\ref{algorithm-optpospos}. Such information is stored as follows.
The array $EXTSCR$ contains the maximum scores of prefixes of the extended
intervals $\insertintoseq{x}{I}{k'}$, for all $k' \in \set{1, 2, \ldots, k}$.
The intervals with best prefix scores obtained so far are kept in the queue
$INTQ$. $Q$ is the rear of the queue $INTQ$, initialized at 0.

\begin{algorithm}[htb]
\caption{Second phase for the case $x > 0$ \label{algorithm-optpospos}}
\KwIn{Arrays $SN$, $INTSCR$, and $XSCR$ computed in the first phase}
\KwOut{An optimal insertion interval for $A$}
\BlankLine

\assignment{$k$}{1}{}
\assignment{$EXTSCR[k]$}{$\max \set{INTSCR[k], XSCR[k]}$}{} \nllabel{lin:k1}
\assignment{$Q$}{$1$}{}
\assignment{$INTQ[Q]$}{$k$}{}
\For {\assignment[nosemicolon]{$k$}{$2, \ldots, \ell$}{}} { \nllabel{lin:S}
\assignment{$DIST$}{$x + SN[INTQ[Q]] - SN[k]$}{} \nllabel{lin:dist}
\While {$DIST \geq 0$ and $DIST + INTSCR[k] > EXTSCR[INTQ[Q]]$} {
\nllabel{lin:while}
	\assignment{$EXTSCR[INTQ[Q]]$}{$DIST + INTSCR[k]$}{} \nllabel{lin:scrinc}
	\If {$Q > 1$ and $EXTSCR[INTQ[Q]] \geq EXTSCR[INTQ[Q-1]]$} {
		\assignment{$Q$}{$Q-1$}{}
		\assignment{$DIST$}{$x + SN[INTQ[Q]] - SN[k]$}{}
	}
}
\assignment{$EXTSCR[k]$}{$\max \set{INTSCR[k], XSCR[k]}$}{}
\nllabel{lin:extscrk}
\If {$EXTSCR[k] < EXTSCR[INTQ[Q]]$} {
	\assignment{$Q$}{$Q + 1$}{}
	\assignment{$INTQ[Q]$}{$k$}{} \nllabel{lin:intqk}
}
}
\Return $INTQ[Q]$ \;
\BlankLine
\end{algorithm}

The correctness of the two-phase algorithm stems from the following lemma.

\begin{lemma} \label{lemma-optpospos}
For every iteration $k$ (just before execution of line~\ref{lin:S} of Algorithm~\ref{algorithm-optpospos}), let $I_{k'}$ be an interval and $k''
= INTQ[Q]$. Then, the following conditions hold:
\begin{enumerate}
  \item $EXTSCR[k''] = \valueofseq{\insertintoseq{x}{I}{k''}
  \setminus \subseq{A}{j_k}{j_\ell}}$;
  \label{it:extscr}
  \item if $Q > 1$ and $k'$ appears in $INTQ$ but $k'' \ne k'$, then
  $k' < k''$ and $\valueofseq{\insertintoseq{x}{I}{k''} \setminus
  \subseq{A}{j_k}{j_\ell}} <
  \valueofseq{\insertintoseq{x}{I}{k'} \setminus \subseq{A}{j_k}{j_\ell}}$; and
  \label{it:inqueue}
  \item if $k' < k$ does not appear in $INTQ$, then $k''$ is such
  that $\valueofseq{\insertintoseq{x}{I}{k''} \setminus \subseq{A}{j_k}{j_\ell}}
  \leq \valueofseq{\insertintoseq{x}{I}{k'} \setminus \subseq{A}{j_k}{j_\ell}}$.
  \label{it:notinqueue}
\end{enumerate}
\end{lemma}
\begin{proof}
By induction on $k$. For $k = 1$, condition~\ref{it:extscr} holds trivially
due to line~\ref{lin:k1}, while conditions~\ref{it:inqueue}
and~\ref{it:notinqueue} hold by vacuity.
Let $k > 1$. We need to analyze the changes in $INTQ$. We start with the
intervals that are removed from $INTQ$. At line~\ref{lin:dist},
Observation~\ref{observation-behaviour-of-a-curve} is used to compute the distance between the curves of $\insertintoseq{x}{I}{k''}$ and
$I_k$. If this distance is negative, then $I_k$ is not a subinterval of
$\insertintoseq{x}{I}{k''}$. Otherwise, condition~\ref{it:extscr} of the
induction hypothesis is used in the comparison of
line~\ref{lin:while} and $EXTSCR[k'']$ is updated at line~\ref{lin:scrinc}
according to~\eqref{eq:recur} using Observation~\ref{observation-behaviour-of-a-curve}.
So, condition~\ref{it:extscr} remains valid for $k$ up to this point of the
execution. If $EXTSCR[k'']$ increases (i.e. line~\ref{lin:scrinc} is executed),
then Observation~\ref{observation-comparing-two-curves} and
condition~\ref{it:inqueue} of the induction hypothesis are evocated to remove
$\insertintoseq{x}{I}{k''}$ from the queue respecting
condition~\ref{it:notinqueue} in case a point of $I_k$ in the curve of
$\insertintoseq{x}{I}{k''}$ overcomes that of an interval that preceeds $I_{k''}$. $EXTSCR[k'']$ is updated again according to~\eqref{eq:recur} in order to satisfy condition~\ref{it:extscr}. This procedure is repeated until condition~\ref{it:inqueue} is valid for the
intervals still in $INTQ$.

Finally, lines~\ref{lin:extscrk}--\ref{lin:intqk} correspond to the
insertion in $INTQ$. The maximum score of the prefix of
$\insertintoseq{x}{I}{k}$ containing $I_k$ and $x$ only is updated at
line~\ref{lin:extscrk} and $\insertintoseq{x}{I}{k}$ enters the queue only if
such maximum score is below the maximum score of the prefix of
$\insertintoseq{x}{I}{INTQ[Q]}$ considered so far.
This implies that conditions~\ref{it:inqueue} and~\ref{it:notinqueue} are also
valid for $k$. %\qed
\end{proof}

\begin{theorem}
 The ISS problem can be solved in $O(n)$ time and space.
 \label{thm:iss}
\end{theorem}

\section{Sorting}
\label{sec:sort}

We now turn our attention to the {\sortingproblem} problem. Its hardness is analized considering the following derived problem. 
\begin{paragraph}
{\bf Restricted version of the \sortingproblem\ problem:} we denote by \restrictedsortingproblem\ the restricted version of the 
\sortingproblem\ problem where, for some two positive integers $k$ and $s$,
$n = 4k - 1$, the elements in $A$ are integers bounded by a polynomial function of $k$, $k-1$ elements are negative, every negative element is equal to 
$-s$, every positive element $a_i$ is such that $s/4 < a_i < s/2$, and $\sumofseq{A} = s$.
\end{paragraph}

A consequence of the
fact that sorting a sequence is similar to accommodate the positive elements in
order to create an appropriate partition into intervals leads to the following
result.

\begin{theorem}
The {\restrictedsortingproblem} problem is strongly NP-hard.
\label{thm:nphard}
\end{theorem}

\begin{proof}
By reduction from the \Problemname{3-Partition} decision problem, stated as follows: given $3k$ positive integers 
$a_1, \ldots, a_{3k}$, 
all polynomially bounded in $k$, and a threshold $s$ such that $s/4 < a_i < s/2$ and $\sum_{i=1}^{3k} a_i = ks$, there exist 
$k$ disjoint triples of $a_1$ to $a_{3k}$ such that each triple sums up to exactly $s$? The \Problemname{3-Partition} problem is 
known to be NP-complete in the strong sense \cite{Garey.Johnson.1979}.

Given an instance $C$ of the \Problemname{3-Partition} problem, an instance of
the {\restrictedsortingproblem} problem is defined by an arbitrary permutation $A$ of the
multiset $C'$ obtained from $C$ by the inclusion of $k-1$ occurences of $-s$. 
A solution for the SSS instance is to choose elements of $C$ for each negative element of $C'$, which
gives a partition of $C$. Since $a_i > s/4$, for all $i \in \{ 1, \ldots, 3k \}$,
every sequence of 4 positive elements chosen from $C'$ has value greater than $s$. Thus, $C$ is a ``yes''
instance of the \Problemname{3-Partition} problem if and only if there exists a permutation $A'$ of $A$ such that $\valueofseq{A'} = s$. %\qed
\end{proof}

We show in the sequel that
Algorithm~\ref{alg:paramsorting} is a parametrized approximation algorithm for the SSS problem. Such an algorithm builds a permutation of $A$ keeping the maximum scoring subsequence of all intervals,
except the last one, bounded by the input parameter plus the largest element of
$A$.
For the last interval, the following holds for every sequence $A$.

\begin{observation} \label{observation-last-interval}
If $N = \seq{score(I_1), score(I_2), \ldots, score(I_\ell)}$ is the sequence of
negative elements composed by intervals' scores, then
$\sumofseq{\subseq{N}{\ell-1}{\ell}} = \sumofseq{A} -
\sumofseq{\subseq{N}{0}{\ell-1}}$. Considering that $I_{\ell}$ is a subsequence
of $A$ and that $\sumofseq{\subseq{N}{0}{\ell-1}} < 0$, we conclude that
$\sumofseq{\subseq{N}{\ell-1}{\ell}}$ is a lower bound for $\valueofseq{A}$ at
least as good as $\sumofseq{A}$.
\end{observation}

Algorithm~\ref{alg:paramsorting} gets as input, in addition to the instance $A$
(with size $n$), the parameter $L$, which depends on $M = \max \{ 0,
\max_{a \in A} a \}$. A variable $S$ is used to keep the score of the interval
being currently constructed.
Just after step~\ref{step:assignpos} is executed, it turns out that $L+M \geq S
\geq L$. On the other hand, execution of step~\ref{step:assignneg} leads to $S
\leq L$ or includes all remaining negative elements in $A'$. Moreover, if $S +
\sumofseq{Q} + \sumofseq{R} < 0$, then a new interval $I_k$ is established and
$S$ is incremented by $-\sumofseq{I_k}$ (and becomes 0). A straightforward
consequence is that $\valueofseq{A'} > L+M$ only if step~\ref{step:finalelem} is executed with positive elements of $A$, and this due to the last interval (in the sense of Observation~\ref{observation-last-interval}). This leads to the following result.

\begin{algorithm}[htb]
\caption{{\paramsorting}$(A, L)$}
\label{alg:paramsorting}
\KwIn{an array $A$ of $n \geq 0$ numbers and a parameter $L \geq M$}
\KwOut{an array $A'$ containing a permutation of $A$}
\BlankLine

Let $A'$ be an array of size $n$ \;
Let $A^- \subseteq A$ and $A^+ \subseteq A$ be the sequences of
negative and nonnegative members of $A$, respectively \label{step:splitsets} \;
\assignment{$j$}{1}{}
\assignment{$S$}{0}{}
\While {$A^- \ne \emptyset$ and $A^+ \ne \emptyset$} {
	Let $Q$ be a sequence of elements of $A^+$ such that $L \leq S  + \sumofseq{Q}
	\leq L+M$, if one exists, or $Q = A^+$ otherwise \nllabel{step:positive} \;
	Assign the elements of $Q$ to $A'[j \ldots j+|Q|-1]$ \;
	\assignment{$j$}{$j + |Q|$}{}
	\assignment{$S$}{$S + \sumofseq{Q}$}{}
	\assignment{$A^+$}{$A^+ \setminus Q$}{}	\nllabel{step:assignpos}
	Let $R$ be a minimal sequence of elements of $A^-$ such that $S +
	\sumofseq{R} < L$, if one exists, or $Q = A^-$ otherwise
	\nllabel{step:negative} \;
	\assignment{$S$}{$\max \set{0, S + \sumofseq{R}}{}$}{}
	Assign the elements of $R$ to $A'[j \ldots j+|R|-1]$ \;
	\assignment{$j$}{$j + |R|$}{}
	\assignment{$A^-$}{$A^- \setminus R$}{} \nllabel{step:assignneg} 
}
Assign the elements of $A^- \cup A^+$ to $A'[j \ldots |A^- \cup A^+|-1]$
\nllabel{step:finalelem} \; 
\Return $A'$
\BlankLine
\end{algorithm}

\begin{lemma}
Let $A$ be an instance of the SSS problem, $A'$ be the sequence returned by
the call \paramsorting$(A, L)$, for some $L \geq M$, and $N'$ be the sequence of the $\ell'$ interval scores
of $A'$. Then,
\begin{equation}
\valueofseq{A'} \leq \max \set{L+M,
\sumofseq{\subseq{{N'}}{\ell'-1}{\ell'}} = \sumofseq{A} -
\sumofseq{\subseq{{N'}}{0}{\ell'-1}}}.
\label{eq:scoreAp}
\end{equation}
Moreover, {\sc ParametrizedSorting}$(A, L)$ runs in $O(n)$ time.
\label{lem:parametrized}
\end{lemma}

The key of our approximation algorithm is to provide Algorithm {\paramsorting}
with an appropriate lower bound parameter. The most
immediate one is $L = \max \set{M, \sumofseq{A}}$, which, however, does
not capture the contribution of the negative members of $A$ whose values are smaller than
$-L$ when $A$ contains at least one nonnegative element. In order to circumvent
this difficult case of Lemma~\ref{lem:parametrized}, assume that $A^*$ is an
optimum solution, with $N^*$ being the sequence of $\ell^*$ scores of the corresponding partition into intervals, 
and $OPT = \valueofseq{A^*}$. According
to~\eqref{eq:scoreAp}, we need to find a new value for $L$ such that
$\sumofseq{\subseq{{N'}}{\ell' - 1}{\ell'}} \leq L \leq OPT$, being
$I_{\ell'}$ the last interval of the sequence $A'$ returned by
{\sc ParametrizedSorting}$(A,L)$, with the purpose of having $\valueofseq{A'}
\leq OPT + M$.

\begin{lemma}
Let $x$ be a real number and
\[
b(x) = \sumofseq{A} + \sum_{a \in B_{x}} (- a - x),
\]
where $B_x = \{ a_i \in A \mid a_i < -x \}$ (note that $B_x$ is a multiset).
Then, $x \geq b(x)$ implies $b(x) \leq OPT$.
\label{lem:lb}
\end{lemma}

\begin{proof}
By contradiction, assume that $x \geq b(x)$ and $b(x) > OPT$. Since $x > OPT$, we get $B_x \subseteq B_{OPT}$. In 
addition, Observation~\ref{observation-last-interval} gives
\[
\begin{array}{rcl}
OPT & \geq & \sumofseq{A^*} - \sumofseq{\subseq{{N^*}}{0}{\ell^*-1}} \\
    & \geq & \sumofseq{A} + \sum_{a \in B_{OPT}} (- a - OPT) \\
    & \geq & \sumofseq{A} + \sum_{a \in B_{x}} (- a - x) \\
    & = & b(x),
\end{array}
\]
which contradicts the assumption $b(x) > OPT$. %\qed
\end{proof}

Based on lemmata~\ref{lem:parametrized} and~\ref{lem:lb}, we define the two-phase Algorithm {\approxsorting}. Its first phase consists in determining the largest $B_L$
satisfying $L \geq M$ (Lemma~\ref{lem:parametrized}) and $L \geq b(L)$ (Lemma~\ref{lem:lb}). Set 
$L_0 = \max \set{M, \sumofseq{A}}$ and take the
elements of a decreasing sequence $P$ on the set $\set{-a \mid a \in A,
a < -L_0} \cup \{ L_0 \}$ (note that, by definition, all elements of $P$ are
distinct). Write this sequence as 
$P = \seq{p_0, p_1, \ldots, p_{|P|-1}}$, which means that $B_{p_0} = \emptyset$ and $b(p_0) = \sumofseq{A}$. 
Then, find the maximal index $i$ (in the range from 0 to $|P|-1$) such that $i = 0$ or $b(p_i) < p_{i-1}$. It is worth mentioning that we can have $b(p_i) < L_0$ when
$\sumofseq{A} < M$. 

The second phase is simply a call {\sc ParametrizedSorting}$\left(A, L = \max \{ L_0, b(p_i) \} \right)$.

\begin{theorem}
{\sc ApproxSorting} is a 2-approximation algorithm for the SSS problem and a $3/2$-approximation algorithm for the
\restrictedsortingproblem\ problem which runs in $O(n \log n)$ time.
\end{theorem}

\begin{proof}
First we show that $L = b(p_i) \leq OPT$ (the case $L = L_0$ is trivial). If $P = \seq{0}$, then $M = 0$ and $OPT = 0$. In this case, $L = OPT = 0$. Otherwise, there are
two subcases. If $i = 0$, then $b(p_0) = \sumofseq{A} \leq L_0 \leq OPT$. On the other hand, if $i > 0$, then, by Lemma~\ref{lem:lb}, 
it is sufficient to show that $b(p_i) < p_{i-1}$
yields $p_i \geq b(p_i)$ or $b(p_i) = b(b(p_i))$. This implication holds since if $p_i < b(p_i) < p_{i - 1}$, then $B_{p_i} = B_{b(p_i)}$.

Let $A'$ be the permutation of $A$ produced by {\sc ApproxSorting}$(A)$, with partition into intervals
$\seq{I_1', I_2', \linebreak[0] \ldots, \linebreak[0] I_{\ell'}'}$. Since each interval $I'_k$ having $\sumofseq{I'_k} < 0$ has an $a \in
B_L$ as last element, we get $\sumofseq{I'_k} \geq a+L$. It turns out that $\sumofseq{\subseq{{N'}}{0}{\ell'-1}} \geq 
\sum_{a \in B_{L}} (a + L) = score(A) - b(L) \geq score(A) - L$.
Observation~\ref{observation-last-interval} leads to
$\sumofseq{\subseq{{N'}}{\ell'-1}{\ell'}} \leq L$. Therefore,
Lemma~\ref{lem:parametrized} gives that $\valueofseq{A'} \leq L + M$.

The approximation factors stem directly from Lemma~\ref{lem:parametrized} and $M \leq \valueofseq{A}$. In special, for
the \restrictedsortingproblem\ case, the first phase of Algorithm {\approxsorting} obtains $L_0 = \max \{ \sumofseq{A} = s, M < s/2 \} = s$
and, by the definition of $B_s$, $L = s$. This leads to the approximation factor $(L+M)/L < (s + s/2)/s = 3/2$. 

Finally, the time complexity is due to the construction of the sequence $P$ (notice that the search for $p_i$ in $P$ can be easily done in linear time). %\qed
\end{proof}

A final remark that can be made in connection with algorithm {\approxsorting} is that the approximation factor of 2 is tight. To see this, consider $x > 0$ and $x/2 < y < x$. The sequence $A$ returned by the call \paramsorting$(\seq{y,-x, y,-x, x}, x)$ is either \seq{y, y,-x, x,-x}, or \seq{y, x,-x, y,-x}, or \seq{x,-x, y, y,-x}. It follows that $2y \leq \valueofseq{A} \leq y+x$. Then, since $OPT = x$, $\frac{\valueofseq{A}}{OPT} \tendsto 2$ as $x-y \tendsto 0$.

\section{Concluding remarks}
\label{sec:conc}

We motivated two problems related to maximum scoring subsequences
of a sequence, namely the \Problemname{Insertion in a Sequence with scores
(\insertionproblem)} and \Problemname{Sorting a sequence by scores (\sortingproblem)} problems. For the ISS problem, we presented a linear time solution, and for the SSS one 
we proved its NP-hardness (in the strong sense)
 and gave a 2-approximation algorithm.

The SSS problem is also closely related to another set partitioning problem,
called {\kpartitionproblem} problem, stated as follows: given a multiset $C$ of positive integers and a positive integer $m$, find a partition of $C$ into $m$ subsets $C_0, C_1, \ldots, C_{m-1}$ such that $\max_{i \in \set{
0, 1, \ldots, m-1 }} \set{\sum_{a \in C_i} a}$ is minimized. Not surprisingly,
given an instance $(C,m)$ of {\kpartitionproblem}, an instance of the
{\sortingproblem} problem can be defined as an arbitrary permutation $A$ of the
multiset $C'$ obtained from $C$ by the inclusion of $m-1$ occurrences of the
negative integer $- {\sumofseq{C}} -1$, indicating that a solution of the
{\sortingproblem} problem for $A$ induces a solution of the {\kpartitionproblem}
problem for $(C,m)$. This problem admits a polynomial time approximation scheme
(PTAS)~\cite{Hochbaum.Shmoys.87,Kellerer.Pferschy.Pisinger.04} as well as list
scheduling heuristics producing a solution which is within a factor of $2 - 1/n$
(being $n$ the number of elements in the input multiset $C$) from the
optimal~\cite{Graham.69}. On the other hand, MAX-3-PARTITION, the optimization version of the problem used in 
the proof of Theorem~\ref{thm:nphard}, is
known to be in APX-hard~\cite{Feldmann.Foschini.12}. A natural open question is, thus, whether there exist a
polynomial time approximation algorithm with factor smaller than 2 for the SSS problem. In this regard, note that although transferring our
approximation factor from the SSS problem to the {\kpartitionproblem} problem is easy, the converse appears harder to be done,
since we do not know in advance how many intervals there should be in an optimal permutation $A'$ of $A$.

\section*{Acknowledgements}

We would like to thank Prof. Siang Wun Song for prolific
initial discussions on the maximum sum subsequence problem.

\bibliographystyle{unsrt}
\bibliography{./bibliography}
\end{document}